\documentclass[pra,aps,twocolumn,groupaddresses,superscriptaddress]{revtex4}
\usepackage[colorlinks=true, allcolors=blue]{hyperref}
\usepackage{graphicx}
\usepackage{amsmath,amssymb,bm}
\usepackage{mathtools}

\usepackage{latexsym}
\usepackage{verbatim}
\usepackage{times}
\usepackage[caption=false]{subfig}
\usepackage{xcolor}
\usepackage{algorithm}
\usepackage[noend]{algpseudocode}
\newtheorem{theorem}{Theorem}
\newenvironment{proof}[1][Proof]{\noindent\textbf{#1.} }{\ \rule{0.5em}{0.5em}}

\def\diag{\operatorname{diag}}

\newcommand{\abs}[1]{\left\vert#1\right\vert}

\newcommand{\s}{\sigma}

\newcommand{\mb}[1]{{\bm#1}}
\newcommand{\tr}{\operatorname{Tr}}

\newcommand{\beq}{\begin{equation}}
\newcommand{\eeq}{\end{equation}}
\newcommand{\baq}{\begin{eqnarray}}
\newcommand{\eaq}{\end{eqnarray}}

\newcommand{\brac}[1]{\left\lbrace #1
	 \right\rbrace}

\def\ket#1{| #1 \rangle}
\def\bra#1{\langle #1 |}

\def\ave#1{\langle #1 \rangle}

\newcommand{\sbrac}[1]{\left[#1\right]}
\newcommand{\pbrac}[1]{\left(#1\right)}

\def\U{\mathcal{U}}

\def\C{\mathcal{C}}

\def\G{\mathcal{G}}

\usepackage[normalem]{ulem}

\newcommand{\hq}[1]{{#1}}

\begin{document}
	
	\title{Regimes of Classical Simulability for Noisy Gaussian Boson Sampling}
	\author{Haoyu Qi}
	\affiliation{Xanadu, 777 Bay St, Toronto ON, M5G 2C8, Canada,}
	\author{Daniel J. Brod}
	\affiliation{Instituto de F\'{i}sica, Universidade Federal Fluminense, Niter\'{o}i RJ, 24210-340, Brazil}
	\author{Nicol\'{a}s Quesada}
	\affiliation{Xanadu, 777 Bay St, Toronto ON, M5G 2C8, Canada,}

	\author{Ra\'{u}l Garc\'{i}a-Patr\'{o}n}
	\affiliation{Centre for Quantum Information and Communication, \'Ecole polytechnique de Bruxelles,
		CP 165, Universit\'{e} libre de Bruxelles, 1050 Brussels, Belgium}

\begin{abstract}
As a promising candidate for exhibiting quantum computational supremacy, Gaussian Boson Sampling (GBS) is designed to exploit the ease of experimental preparation of Gaussian states. 
However, sufficiently large and inevitable experimental noise might render GBS classically simulable.
In this work, {we formalize} this intuition by establishing a sufficient condition for approximate {polynomial-time classical} simulation of noisy GBS --- {in the form  of} an inequality between the input squeezing parameter, the overall transmission rate and  the quality of photon detectors. Our result serves as a non-classicality test that must be passed by any quantum \hq{computational} supremacy demonstration based on GBS. We show that, for most linear-optical architectures, where photon loss increases exponentially with the circuit depth, noisy GBS loses its quantum advantage in the asymptotic limit. Our results thus delineate intermediate-sized regimes where GBS devices might considerably outperform classical computers for modest noise levels.
Finally, we find that increasing the amount of input squeezing is helpful to evade our classical simulation algorithm, which suggests a potential route to mitigate photon loss.
\end{abstract}
\maketitle

Quantum computational supremacy ~\cite{harrow2017quantum,Preskill2018quantumcomputingin}---when a quantum device performs a computational
task beyond the capabilities of classical computers---is a long-anticipated milestone towards practical quantum computation. 
A recent publication has 
claimed to have achieved such milestone on a quantum processor with 53 superconducting qubits~\cite{arute2019quantum}. However, subsequent work managed to significantly speed up the classical simulation of 
this multi-qubit 
device~\cite{pednault2019leveraging}. This suggests that quantum computational supremacy, instead of
being a one-shot experimental proof, will be the result of a long-term competition between
continuously improved quantum devices and faster classical simulations. 

{Since near-term quantum devices do not benefit from error correction, in a realistic scenario noise increases with the size of the experiment and potentially negates its quantum advantage. Numerous works have studied the effects of noise on Boson Sampling (BS), including partial photon distinguishability~\cite{rohde2015boson,shchesnovich2015tight,renema2018efficient}, fabrication imperfections ~\cite{kalai2014gaussian,leverrier2013analysis,arkhipov2015bosonsampling}, losses~\cite{rahimi2016sufficient,oszmaniec2018classical,garcia2017simulating,renema2018quantum,aaronson2016losses}, and detector dark counts~\cite{rahimi2016sufficient}. 
In contrast with BS, no rigorous analysis has been carried out for noisy Gaussian Boson Sampling (GBS), a variant of BS which finds applications in specific computational problems such as dense subgraph searching~\cite{arrazola2018using,banchi2019molecular}, perfect matching counting~\cite{bradler2018graph}, graph isomorphism~\cite{bradler2018gaussian} and the simulation of vibrational spectra~\cite{huh2015boson}.
In addition, GBS is related to non-Gaussian probabilistic state  engineering~\cite{sabapathy2018near,su2019generation,gagatsos2019efficient}.}

Concretely, GBS~\cite{hamilton2017gaussian} describes a computational task in which photon statistics is directly measured from a Gaussian state. 
An arbitrary GBS instance can be implemented by: (i) deterministic preparation of
$K$ single-mode squeezed vacuum states; 
(ii) interference over an $M$-mode interferometer~\cite{reck1994experimental,clements2016optimal}; (iii) sampling of output statistics by photon number resolving detectors. 
The deterministic sources, together with its high generation probability and sampling rate,
render GBS a {promising} alternative to other proposals.

In this work, we take the first step towards answering the following questions: what is the minimum squeezing needed to demonstrate quantum \hq{advantage} with GBS? How much photon loss would render the device classically simulable? What level of dark counts and inefficiencies can be tolerated in the detectors? These questions are crucial for a deeper understanding of the origin of quantum complexity in GBS, as well as for the design of any experiment aiming at a demonstration of quantum \hq{computational} supremacy.

\textit{Main result and ramifications} --- 
Our main result states that a noisy GBS device (see Fig.~\ref{fig:noise-model}(a) for a schematic) can be classically efficiently simulated up to error $\epsilon$ if the following condition is satisfied:
\begin{align}\label{eq:main-result}
\text{sech}\sbrac{\frac{1}{2}\Theta\pbrac{\ln\pbrac{\frac{1-2q_D}{\eta e^{-2r}+1-\eta}}}} > e^{-\frac{\epsilon^2}{4K}}~.
\end{align}
Here the photon source is characterized by the squeezing parameter $r$, and $\eta$ is the overall transmission rate. The quality of the photon detectors is quantified by $q_D:=p_D/\eta_D$, where $p_D$ is the dark count rate and $\eta_D$ is the quantum efficiency. \hq{Finally, $\Theta$ is the ramp function $\Theta(x):=\max(x,0)$.}

\hq{We briefly summarize our classical algorithm as follows:
\begin{enumerate}
    \item If inequality \eqref{eq:main-result} holds, proceed; otherwise exit the algorithm.
    \item Calculate the classical Gaussian state which is closest to the output state of the noisy device. We derive an efficient analytical formula for this procedure.
    \item  Sample from the classical Gaussian state according to the method given in Ref.~\cite{rahimi2016sufficient}, which can be done efficiently.
\end{enumerate}
See Sec.~VI of the Supplemental Material \cite{supp} for a detailed description of our algorithm, which includes Refs.~\cite{arkhipov2012bosonic,bremner2010classical,Lund2014SBS,Scottblog,Brod2015depth,marian2002quantifying,spedalieri2012limit,muller2013on,wilde2014strong,frank2013monotonicity,beigi2013sandwiched,van2014renyi,seshadreesan2018renyi,wilde2017gaussian,killoran2019strawberry,Qirepo}.}

{A finite-sized experiment will count as strong evidence towards quantum \hq{computational} supremacy when the best known classical algorithm to simulate it, running on a classical supercomputer, takes a reasonably large amount of time (for further discussion, see \cite{bremner2016average,neville2017classical,boixo2018characterizing}). Therefore, inequality \eqref{eq:main-result}, for reasonably small $\epsilon$, defines a natural and useful test that any alleged quantum \hq{advantage} demonstration with noisy GBS must pass.}

To illustrate this, we test inequality \eqref{eq:main-result} on several recent small-scale GBS experiments \cite{clements2018approximating,paesani2018generation,zhong2019experimental}. For instance, in Ref.~\cite{paesani2018generation}, $K=4$ squeezed vacuum states with $r\approx 0.1$ were input into a 12-mode random-walk circuit with overall transmission $\eta=0.088$ and detector efficiency $\eta_D=0.78$. For typical superconducting nanowire single photon detectors, the dark count rate is  $p_D\approx10^{-4}$~\cite{hadfield2005single}. Using these numbers, the experiment can be efficiently simulated by our algorithm with error $2\%$. A simple python implementation takes $2$ms to output a sample on a laptop . In the most recent GBS demonstration, $K=6$ squeezed vacuum states with an average  $r\approx 0.38$ are coupled into a 12-mode interferometer implemented with bulk optics in free space~\cite{zhong2019experimental}. With overall transmission $\eta=0.89$ \footnote{The transmission rate of the interferometer  is 0.99, which can be found in in Ref.~\cite{zhong201812}. We assume that the efficiency of coupling the output light to the detector is $0.9$, which is reported from a recent BS experiment~\cite{wang2019boson}. This gives an overall transmission around $0.89$}, $\eta_D=0.82$ and assuming $p_D=10^{-4}$, 
we find that inequality  \eqref{eq:main-result}  has no solution for any $\epsilon\in [0,1]$, which makes this experiment pass our classicality test. 

We can also study the behavior of inequality \eqref{eq:main-result} in the asymptotic limit. \hq{Assuming that $q_D$ and $r$ are constants as $K$ increases, in the asymptotic limit of large $K$, we have the following} condition for efficient classical simulability
\begin{align}\label{eq:asymptotic}
\eta < \eta_\infty + \frac{1-2q_D}{1-e^{-2r}}\frac{\epsilon}{\sqrt{2K}}+\omega\pbrac{\frac{1}{\sqrt{K}}}~.
\end{align}
Here $\eta_\infty := q_D(1+\coth r)$, and $\eta <\eta_\infty$ {is the corresponding condition that follows for exact sampling ($\epsilon=0$) obtained previously in} Ref.~\cite{rahimi2016sufficient}. See Fig.~\ref{fig:suff-cond}(a) and (b) for plots of this inequality for various squeezing levels, transmission rates, and detector qualities. 

From our result cast into the form \eqref{eq:asymptotic}, we observe the following: (i) by including some error tolerance in the classical simulation, our result improves on previous algorithms of Ref.~\cite{rahimi2016sufficient} and gives tighter bounds {on noise parameters for finite-size experiments.} (ii) our result allows more general types of noise, including a combination of photon loss and dark counts, compared to previous analyses for standard BS~\cite{garcia2017simulating,oszmaniec2018classical}; (iii) increasing the input squeezing reduces the upper bound on the transmission rate, which underpins the notion of squeezing as a non-classical resource.

An important scenario which allows comparison with previous BS results is that of pure photon loss
by assuming perfect detectors, $q_D=0$. In this case, for exact sampling $(\epsilon=0)$, inequality \eqref{eq:main-result} gives trivial result $\eta<0$---a major drawback of the methods in Ref.~\cite{rahimi2016sufficient}. This can be understood by observing that a squeezed state, under the effect of pure loss, never becomes exactly a classical state (though it approximates one arbitrarily well). 

{In Sec.II of the Supplemental Material~\cite{supp}, we argue that exact sampling of GBS with arbitrary loss is computationally hard. The seemingly counter-intuitive result suggests that the notion of exact simulability is often too stringent to be considered in the error analysis of quantum devices, as even a real-world experiment is incapable of sampling from the exact theoretical distribution. Therefore, below we analyze the effect of pure loss by considering
the more relevant approximate simulation.}

In a realistic setup, the input squeezing (and the energy) per mode does not scale with $K$. Multiplying both sides of inequality~\eqref{eq:asymptotic} by the average photon number $\bar{N}=K\sinh^2 r$, we find that the inequality is satisfied if $\eta \bar{N}<O(\epsilon\sqrt{\bar{N}})$. Thus, lossy GBS can be efficiently simulated when the average number of surviving photons is less than $O(\sqrt{\bar{N}})$, which matches the scaling obtained for standard BS~\cite{oszmaniec2018classical,garcia2017simulating}.

In most linear-optical architectures, photon loss is defined by unit depth of the circuit, leading to
an exponential decrease of transmission with the depth. Our results imply that, {asymptotically,} GBS implemented on these platforms is rendered efficiently simulable by losses if the depth is linear in the number of modes.  
It is easy to see that condition \eqref{eq:asymptotic} is always satisfied for these circuits when they have super-logarithmic depth, which also holds for BS \cite{garcia2017simulating,oszmaniec2018classical}. 
In the other extreme, for planar circuits (i.e.,\ with only nearest-neighbor beam splitters) of \emph{logarithmic} depth we can construct a tensor network simulation that runs in quasi-polynomial time. This is analogous to a similar algorithm for BS \cite{garcia2017simulating}. The only difference is that we need to introduce an additional cutoff on the Hilbert space, for large photon numbers, that does not degrade too much the Gaussian state nor slows down the simulation.

From this asymptotic analysis we conclude that lossy GBS inevitably loses its quantum advantage for sufficiently large sizes, {at least within most current architectures}. Nonetheless, a large running time advantage might be still be possible for GBS devices with intermediate sizes and modest (but non-negligible) noise rates. A similar situation has been discussed in the context of qubits~\cite{gao2018efficient}.

\textit{Gaussian states and classicality}--- Our results crucially follow from the fact that classical Gaussian states can be sampled efficiently~\cite{rahimi2016sufficient}. Here we briefly review relevant definitions and a few well-known facts.

Consider a single-mode bosonic system and let $\hat{\mb{x}}:=(\hat{q},\hat{p})$ denote the row vector of its quadrature operators. They satisfy the commutation relations ($\hbar=2$)
$
\sbrac{\hat{x}_j,\hat{x}_k} =2 i\bm{\Omega}_{jk}~,
$
where $
\mb{\Omega}
=\left( \begin{smallmatrix}
0 & -1 \\ 1 & 0
\end{smallmatrix} \right)$. A state $\rho$ is Gaussian if it is fully characterized by its mean vector $\bar{\mb{x}}_\rho=\tr\brac{\rho \hat{\mb{x}}}$ and its covariance matrix $
\bm{V}_\rho^{j,k} = \frac{1}{2}\tr{}\brac{\rho\sbrac{\hat{x}_j-\bar{x}_\rho^j,\hat{x}_k-\bar{x}_\rho^k}_+}~,
$
where $\sbrac{\cdot,\cdot}_+$ is the anti-commutator~\cite{weedbrook2012gaussian}. {All distance measures we consider are minimized when two Gaussian states are displaced along the same direction by the same amount~\cite{seshadreesan2018renyi},} so we set $\bar{\mb{x}}=0$ from now on.

The ${(t)}$-ordered phase-space quasi-probability distribution [$({t})$-PQD] of a single mode Gaussian state is given by
\begin{align}\label{eq:PQD-Gaussian}
W^{(t)}_\rho(\mb{x})=\frac{\exp\sbrac{-\frac{1}{2}\mb{x}^T(\mb{V}_\rho-t\mb{I}_2)^{-1}\mb{x}}}{2\pi\sqrt{\det(\mb{V}_\rho-t\mb{I}_2)}}~.
\end{align}
For ${t} =
-1, 0$ and $1$, we obtain the Husimi, Wigner and Glauber-Sudarshan functions, respectively.
Equation (\ref{eq:PQD-Gaussian}) only holds when $\mb{V}_\rho-t\mb{I}_2$ is positive definite. Since the covariance matrix $\mb{V}_\rho$ is positive definite, there always exists $t^*_{\rho} \in [ 0,1 ]$ 
such that: for $t<t^*_{\rho}$ the $(t)$-PQD is a Gaussian function; for $t=t^*_{\rho}$, the $(t)$-PQD has $\delta$-function singularities; and for $t>t^*_{\rho}$, the $(t)$-PQD does not exist. A Gaussian state is called classical if its $P$ function (i.e.,\ for $t=1$)
is well-behaved.
We slightly generalize this notion of classicality and say that $\rho$ is $(t)$-classical if 
{
$\mb{V}_\rho-t\mb{I}_2$ is positive definite, or equivalently, if}
its $(t)$-PQD is non-singular (i.e.,\ if $t\leq t^*_{\rho}$).
Finally, we denote $\C_G^{(t)}$ as the set of $(t)$-classical Gaussian states.

 \begin{figure}[t]
	\centering
	\includegraphics[width=0.45\textwidth]{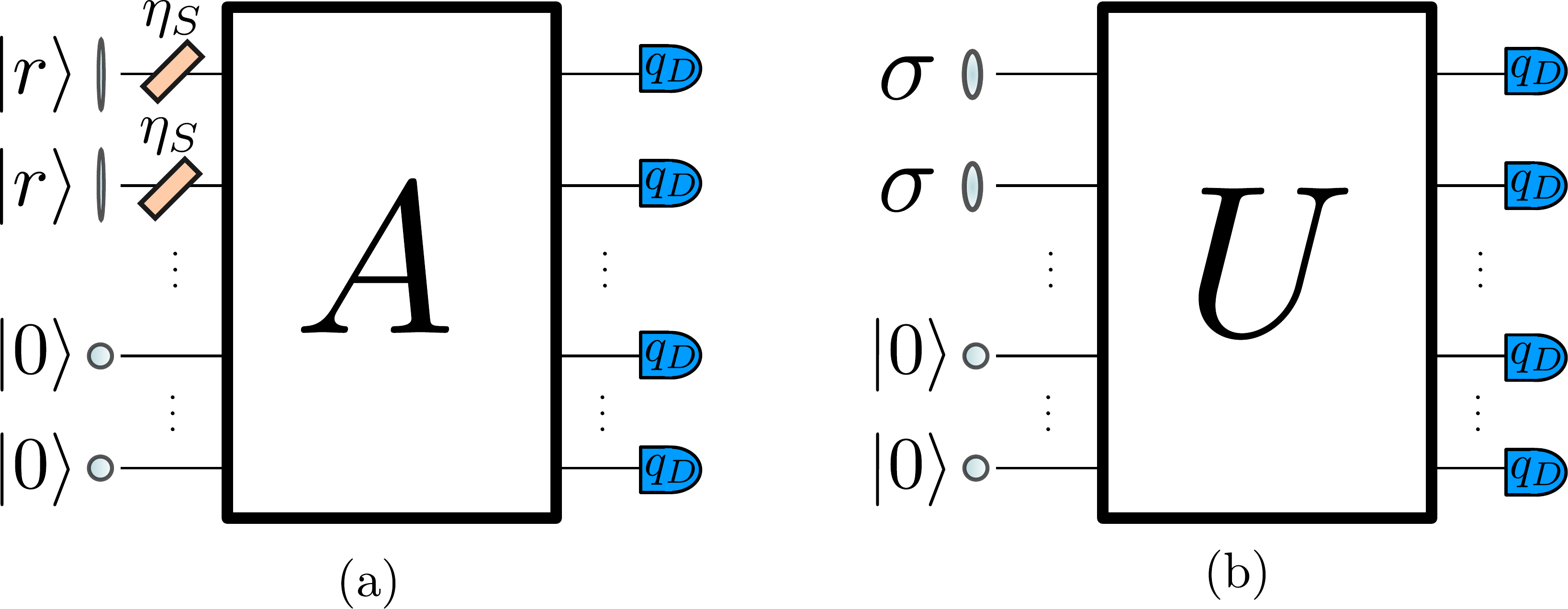}
	\caption{Simplification of noise model. (a) A realistic GBS experiment suffers from imperfect squeezing, modelled by an ideal squeezer followed by a loss channel with transmission $\eta_S$, a lossy interferometer, described by a sub-unitary matrix $\mb{A}$, and inefficient noisy threshold detectors, characterized by $q_D:=p_D/\eta_D$ where $p_D$ and $\eta_D$ are the dark count rate and its quantum efficiency, respectively. (b)  Under the assumption of uniform loss within the interferometer, we can absorb all losses into the thermalized squeezed state $\sigma$. \hq{We also assume that $r$ and $q_D$ do not change as the device size increases, which is reasonable for the regime of near-term experiments.}\label{fig:noise-model}}
\end{figure}
\textit{Noise model}---The input squeezed light
is usually noisy due to photon loss in its preparation~\cite{vernon2018scalable}. We model such source imperfection as an ideal squeezed vacuum state followed by a lossy channel with transmission $\eta_S$. An $M$-mode lossy interferometer is mathematically described by a sub-unitary transformation matrix $\mb{A}$. By assuming uniform loss, which is usually a good approximation for integrated platforms, the {interferometer} can be simplified to $M$ single-mode lossy channels, each with transmission $\eta_I$, followed by an ideal unitary transformation $\mb{U}$ (see \cite{garcia2017simulating,oszmaniec2018classical} for a rigorous treatment). That is, we can write $\mb{A}=\eta_I \bm{U}$. As a final simplification, we combine the source and interferometer losses into a single loss channel with transmission $\eta:=\eta_S\eta_I$ and encapsulate it into a mixed input Gaussian state $\sigma$ (see Fig.~\ref{fig:noise-model}).

{The current proofs for hardness of (Gaussian) Boson Sampling require a no-collision regime}~\cite[Sec.~I]{supp}, where the probability of detecting more than one photon per mode is negligible. In this regime  we can replace photon counting detectors by threshold detectors without loss of generality~\cite{rahimi2016sufficient,quesada2018gaussian}.
Each threshold detector has sub-unity quantum efficiency denoted by $\eta_D$ and registers random dark counts with probability $p_D$. Following Ref.~\cite{rahimi2016sufficient,barnett1998imperfect} we describe {them} by the POVM elements: $\Pi_0=(1-p_D)\sum_{n=0}^\infty(1-\eta_D)^n\ket{n}\bra{n}$ and $\Pi_1=I-\Pi_0$. We define $q_D=p_D/\eta_D$ as the figure of merit that characterizes the detectors.

With this noise model, the probability distribution at the output reads $P(\mb{n})=\tr\brac{\rho_{\text{out}}\Pi_{\mb{n}}}$, with $\mb{n}=(n_1,n_2,\ldots, n_M)$ and $n_i \in \{0,1\}$, where
\begin{align}\label{eq:rho-out}
\rho_{\text{out}}=\mathcal{U}\left(\sigma^{\otimes K}\otimes \ket{0}\bra{0}^{\otimes (M-K)}\right)\mathcal{U}^\dagger~,
\end{align}
and $\mathcal{U}$ is the Hilbert space unitary associated to the unitary matrix $\bm{U}$ acting in mode space.
{Proving the conditions for efficient classical sampling from these distribution, namely inequality \eqref{eq:main-result}, is the subject of the remainder of this paper}.

\textit{Approximate simulation of noisy GBS}---In Ref.~\cite{rahimi2016sufficient}, the authors studied the problem of exact sampling from an $M$-mode quantum state of the form
\begin{align}\label{eq:rho-tilde}
\tilde{\rho}=\mathcal{U}\pbrac{\tau^K\otimes\ket{0}\bra{0}^{M-K}}\mathcal{U}^\dagger
\end{align}
using threshold detectors of quality $q_D$. They proved this task can be simulated exactly with polynomial running time if  $\tau$ is a $(t)$-classical state, for $t\in[1-2q_D,1]$ (See \cite[Sec. VI]{supp} for a brief review of this result). In this work, we restrict our consideration to $(t)$-classical Gaussian states for a given $q_D$.

{It is natural to expect that, when the mixed input state $\sigma$ is close to some $(t)$-classical Gaussian state $\tau$, the corresponding noisy GBS experiment can be efficiently simulated with small error. Since any such state $\tau$ leads to an efficient classical simulation, we minimize the distance to $\sigma$ over all possible choices of $\tau$. It is this intuition that eventually leads to our main result given by condition~\eqref{eq:main-result}.}

We start our formal derivation by connecting the distance between the two output distributions to the distance between the two output states. The former is usually measured by the total variation distance, $D_T(\tilde{P},P):=\frac{1}{2}\sum_{\mb{n}}\abs{\tilde{P}(\mb{n})-P(\mb{n})}$. For the distance between states we choose the sandwiched R\'{e}nyi relative entropy due to its generality, denoted by $D_\alpha(\sigma\Vert\tau)$, 
for $\alpha\in[\frac{1}{2},1)$~\cite{muller2013on,wilde2014strong}. The two measures can be connected by the following inequality chain:
\begin{align}\label{eq:bound-D_T}
D_T(P,\tilde{P})\leq \sqrt{\frac{2}{\alpha}D_\alpha(P\Vert \tilde{P})}\leq \sqrt{\frac{2}{\alpha}D_\alpha(\rho\Vert \tilde{\rho})}~.
\end{align}
We first used a generalized Pinkser's inequality and, in a slight abuse of notation, $D_\alpha(P\Vert \tilde{P})$ also denotes the R\'{e}nyi divergence between two distributions~\cite{van2014renyi}. The second inequality follows from the data-processing inequality under quantum measurements~\cite{muller2013on}.

The last quantity can be greatly simplified due to the similar structures of Eqs.~\eqref{eq:rho-out} and \eqref{eq:rho-tilde}. Since the sandwiched R\'{e}nyi relative entropy is invariant under the action of the unitary transformation and is additive under tensor products~\cite{muller2013on}, we have that $D_\alpha(\rho\Vert \tilde{\rho})=KD_\alpha(\sigma\Vert\tau)$. Substituting this into  Eq.~\eqref{eq:bound-D_T} we see that $D_T(P,\tilde{P})\leq\epsilon$ whenever 
\begin{align}\label{eq:suff-cond-alpha}
\frac{2}{\alpha}D_\alpha(\sigma\Vert\tau)< \frac{\epsilon^2}{K}~.
\end{align} 

We could optimize this bound over $\alpha\in \left[\tfrac{1}{2},1\right)$ --- since $D_\alpha$ is non-decreasing over $\alpha$~\cite{muller2013on}, we expect an optimal $\alpha^*$ exists. For simplicity we only consider the case $\alpha=\tfrac{1}{2}$, which admits analytical calculations; a numerical optimization over $\alpha$ can be found in~\cite[Sec.~V]{supp}.
In this special case, we have $D_{\frac{1}{2}}(\sigma,\tau)=-\ln F(\sigma,\tau)$, where $F(\sigma,\tau):=\tr\brac{\sqrt{\sqrt{\sigma}\tau\sqrt{\sigma}}}^2$ is the quantum fidelity. As we mentioned earlier, we will have a tighter bound if  $\tau$ is further optimized over $C^{(t)}_G$ with $t\in [1-2q_D,1]$. We thus have
\begin{align}\label{eq:sufficient-condition-fid}
-\ln\sbrac{F_{\max}(\eta,q_D)}< \frac{\epsilon^2}{4K}~,
\end{align}
with $F_{\max}(\eta,q_D):=\max_{\tau\in \C_{G}^{(t)}, t\in[1-2q_D,1]}F(\sigma,\tau)$. 

The argument above shows that a noisy GBS can be efficiently simulated with error no more than $\epsilon$ if the above inequality holds; in other words, we have reduced our multi-mode problem into a single-mode optimization.

\textit{Optimization of the fidelity}---To perform the optimization, it is convenient to use the squeezed thermal states (STS) parameterization~\cite{marian1993squeezed}. Any single-mode Gaussian state $\rho$ can be written as 
$\rho=S(s_\rho,\phi_\rho)\rho_T^{n_\rho}S^\dagger(s_\rho,\phi_\rho)$. Here  $\rho_T^{n_\rho}$ is a thermal state with average photon number ${n_\rho}$, and 	
$S(s_\rho,\phi)=\exp\sbrac{\frac{1}{2}s_\rho e^{i\phi_\rho}\hat{a}^{\dagger 2}-\frac{1}{2}s_\rho e^{-i\phi_\rho}\hat{a}^{2}}$ is the squeezing operator with squeezing parameter $s_\rho$
and phase rotation $\phi_\rho$. For the lossy squeezed state $\sigma$ appearing in our noise model (Fig.~\ref{fig:noise-model}), according to Ref.~\cite{supp}, Sec.~II, we have $\phi_\sigma=0$ and
\begin{align}\label{eq:sigma-s}
s_\sigma&=\frac{1}{4}\ln\frac{a_+}{a_-}~,\\\label{eq:sigma-n}
n_\sigma&=\frac{1}{2}(\sqrt{a_+a_-}-1)~.
\end{align}

A straightforward calculation  using an analytical expression of the quantum fidelity  between two single-mode Gaussian states ~\cite{holevo1975some} leads to
\begin{align}
F &= \frac{1}{\sqrt{\Delta + \Lambda} - \sqrt{\Lambda}}, \\
\Delta&=(n_\sigma-n_\tau)^2+(2n_\sigma+1)(2n_\tau+1)\cosh^2(s_\sigma-s_\tau)~, \nonumber\\
\Lambda&=4n_\sigma(n_\sigma+1)n_\tau(n_\tau+1). \nonumber 
\end{align}

Since $\tau$ is a $(t)$-classical Gaussian state, it is straightforward to show that this implies $s_\tau\leq \frac{1}{2}\ln\frac{2n_\tau+1}{t}$~\cite[Sec.~III]{supp}. 
After solving this constrained optimization problem, we obtain that the maximum fidelity is
$
\text{sech}\sbrac{\Theta\pbrac{s_\sigma-\frac{1}{2}\ln \frac{2n_\sigma+1}{t}}}
$ ~\cite[Sec.~IV]{supp}. Finally, we optimize this quantity again over $t\in[1-2q_D,1]$. Clearly, the maximum value is achieved at $t = 1-2q_D$. Using Eqs.~\eqref{eq:sigma-s}\--\eqref{eq:sigma-n} we finally arrive at the following solution:
\begin{align}\label{eq:d_min}
F_{\max}(\eta,q_D)=\text{sech}\sbrac{\frac{1}{2}\Theta\pbrac{\ln\pbrac{\frac{1-2q_D}{\eta e^{-2r}+1-\eta}}}}~.
\end{align}
Plugging this into Eq.~\eqref{eq:sufficient-condition-fid} we attain our main result for classical simulability of noisy GBS, expressed in Eq.~\eqref{eq:main-result}.
This inequality draws a boundary in the parameter space $(\eta, q_D, K, \epsilon)$. In Fig.~\ref{fig:suff-cond} (a) and (b) we plot them in the $\eta-K$ plane for several values of the squeezing parameter $r$ and detector quality $q_D$. For noise parameters below the corresponding solid line, GBS can be efficiently simulated with error $\epsilon\leq 0.01$. 
{We observe that} the region where quantum computational supremacy has not been ruled out expands as we increase the quality of detectors or the amount of input squeezing; And the latter underpins the notion of squeezing as a non-classical resource.
For bounds obtained by using general sandwiched R\'{e}nyi relative entropies, our numerical simulation results shown in Fig.~\ref{fig:suff-cond}(c) and (d) suggest that the quantum fidelity (corresponding to $\alpha=1/2$) is optimal, while quantum relative entropy (corresponding to $\alpha\rightarrow 1$) gives the worst bound.  For more details see the Supplemental Material~\cite{supp}.
\begin{figure}[t]
	\includegraphics[scale=0.3]{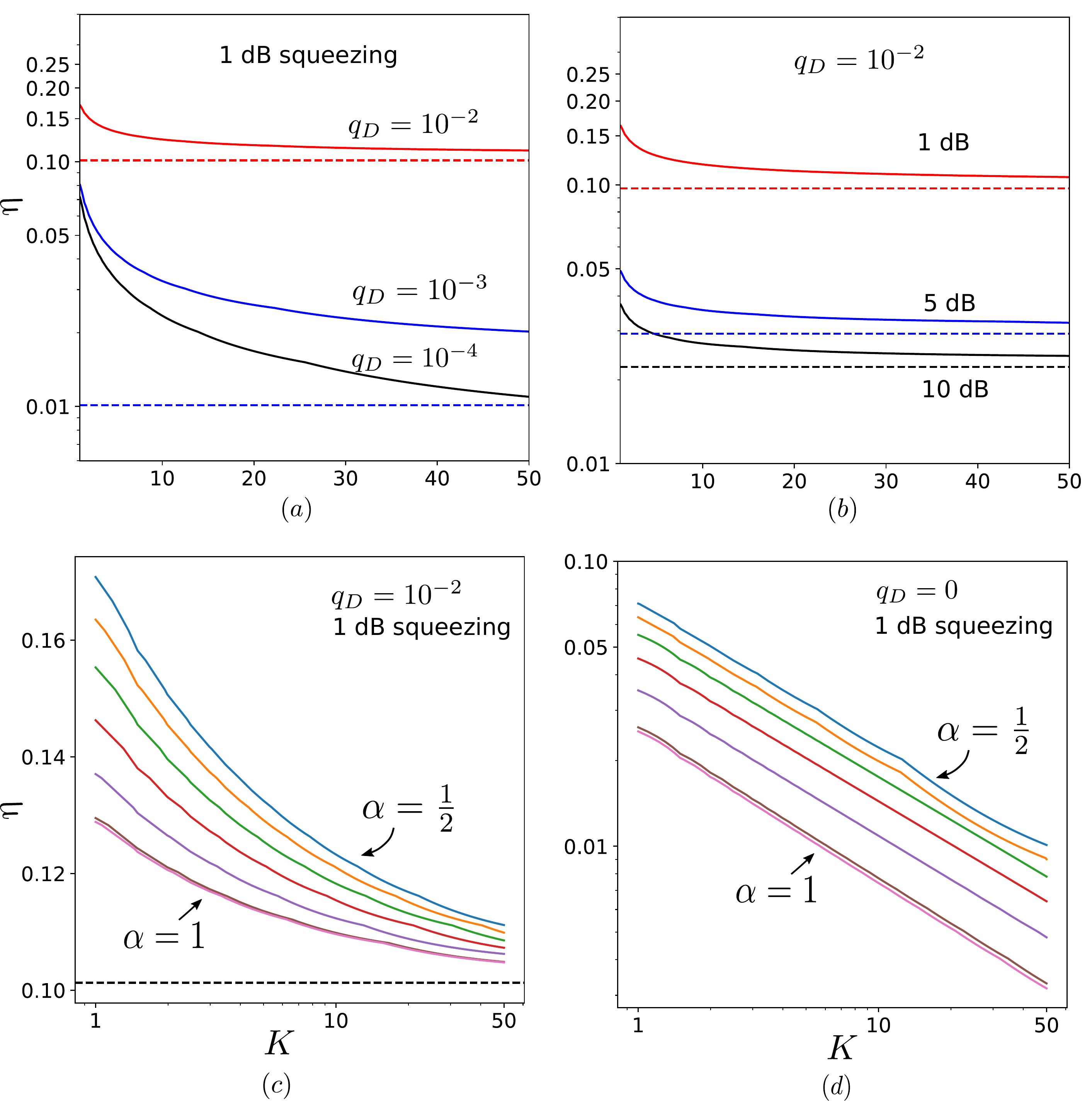}
	\caption{(a), (b) 
	{Sufficient condition for efficient simulation} 
	 of GBS by using the quantum fidelity. In both cases $\epsilon=0.01$. Solid lines are the conditions given by Eq.~\eqref{eq:main-result}. Dashed lines correspond to the conditions given by $\eta=\eta_\infty$. Above such lines GBS can be simulated approximately and exactly in polynomial time, respectively. In (a) we fix input squeezing to 1~dB, and from top to bottom we choose $q_D=10^{-2},10^{-3},10^{-4}$. In (b) we fix $q_D=10^{-2}$ and from top to bottom we choose 1~dB, 5~dB and 10~dB squeezing. (c), (d) 
	{Sufficient condition for efficient simulation} 
	of GBS by using the $\alpha$-order sandwiched R\'{e}nyi relative entropy, as given in Eq.~\eqref{eq:suff-cond-alpha}. From top to bottom $\alpha=0.5,0.6,0.7,0.8,0.9,0.99,0.999$. Each line is obtained by numerically optimizing $D_\alpha(\sigma\Vert\tau)$ over $\tau\in\C_{G}^{(t)}$ for $t\in[1-2q_D,1]$. In both cases we fix 1 dB of squeezing. In (c) we choose $q_D=10^{-2}$ while in (d) we consider perfect detectors. 
	}
	\label{fig:suff-cond}
\end{figure}

{\textit{Conclusion}.---In this work, we established a non-classicality condition that any quantum \hq{computational} supremacy demonstration based on GBS must satisfy, connecting important experimental parameters, such as the squeezing in the photon
source, the overall photon transmission rate, and the dark count rates in photon detection. Our result shows that lossy GBS can be efficiently simulated when the average number of
surviving photons is quadratically related to the number of input photons, which matches
the scaling obtained for BS~\cite{garcia2017simulating,oszmaniec2018classical}. 

Motivated by recent works for BS~\cite{renema2018quantum,renema2018efficient,moylett2019classically}, exploiting the properties of the combinatorial quantity associated to the
statistics of bosons is an important direction of future research, which could lead to
similar improvement on classical simulation of noisy GBS. Another interesting topic for future research is whether non-uniformity of photon loss can be leveraged to improve our simulation \cite{brod2019nonuniform}}.

\begin{acknowledgments}
We thank Juan Miguel Arrazola, Kamil Br\'{a}dler, Ish Dhand, Saleh Rahimi-Keshari, Jonathan Lavoie, Chao-Yang Lu, Jelmer J. Renema, Daiqin Su, Maria Schuld, Krishna K. Sabapathy, Zachary Vernon, Christian Weedbrook, and Han-Sen Zhong for helpful discussions. We also thank the anonymous referees for their useful comments which greatly improve the presentation of our. HQ is grateful to Weiying Tang for a careful reading of the manuscript. DJB acknowledges support from CNPq project Instituto Nacional de Ci\^encia e Tecnologia de Informa\c c\~ao Qu\^antica.
\end{acknowledgments}
\bibliography{ref}

\cleardoublepage
\onecolumngrid
\begin{center}
	{\huge Supplemental materials}
\end{center}

 \section{Non-collision regime for GBS}
For standard boson sampling with $N$ single-photon inputs and $M$ modes, it was shown \cite{aaronson2011computational,arkhipov2012bosonic} that the probability of detecting collision events is bounded as follows:
\begin{align}\label{eq:p-collision}
\ave{P_{\text{collision}}}_\U \leq \frac{8N^2}{M}~,
\end{align}
where the average is over Haar-random unitaries. The proof of Eq.~\eqref{eq:p-collision} relies on the fact that Haar-random unitaries map any $N$-photon, $M$-mode state onto the maximally mixed state (its density matrix is given by the identity on the corresponding Hilbert space). 

For GBS, the input state has indefinite photon number. Specifically, the probability of generating $S$ photon pairs is given by \cite{hamilton2017gaussian}
\begin{align}\label{eq:neg-bin}
F(S) = \binom{\frac{K}{2}+S-1}{S}\text{sech}^K(r)\tanh^S(r)~.
\end{align}
Therefore, from Eq.~\eqref{eq:p-collision}, the probability of detecting collision events at the output of GBS satisfies
\begin{align}
\ave{P_{\text{collision}}}_\U &\leq \sum_{S=0}^\infty F(S)\sbrac{\frac{8(2S)^2}{M}} = \frac{32}{M}\ave{S^2}_F~.
\end{align}

$F(S)$ in Eq.~\eqref{eq:neg-bin} is a negative binomial (or Pascal) distribution. In the large $K$ limit, $F(S)$ converges to a Gaussian distribution with mean value
$
\frac{K}{2}\sinh^2r
$
and variance
$
\frac{K}{2}\sinh^2r\cosh^2r
$, which gives $\ave{S^2}_F=O(K^2)$.
Therefore, we also expect no-collision outputs in GBS to dominate whenever
\begin{align}
M=O(K^2)~.
\end{align}

\section{Exact lossy GBS is hard}
\label{apx:exact-simulation}

In this section we give evidence that exact classical simulation of a lossy GBS device cannot be efficient, unless the polynomial hierarchy collapses to its third level. The post-selection based argument we use is standard and was used to prove similar claims for many different restricted models of quantum computation (see e.g.\ \cite{bremner2010classical,aaronson2011computational}), so we only detail the parts of the argument that pertain to GBS. The construction we use is directly inspired by the scattershot boson sampling model \cite{Lund2014SBS,Scottblog}, though our purpose is different, as we are interested especially in the effect of losses and the complexity of the model. 

For this proof, we assume that losses are uniform within the interferometer $\bm{U}$, and so we can move all losses to the end (this is a standard assumption that is a good approximation for e.g.\ integrated photonic devices, but was also shown to hold under more general conditions \cite{oszmaniec2018classical}). In contrast to the results in the main paper, here we also ignore all other sources of losses. This is an important caveat, but can be justified as losses in photon sources and detectors are effectively constant, whereas losses inside the interferometer $\bm{U}$ scale with its depth (which, for boson sampling, also typically scales with the number of photons). Therefore, photon loss within the linear optical network is the main scalability botteleneck. We leave it as an open question whether this caveat can be eliminated.

\begin{theorem}\label{thm:exact-GBS}
	If there is an efficient classical algorithm to sample from the output distribution of a lossy Gaussian boson sampling instance exactly (or up to multiplicative error), then the polynomial hierarchy collapses to its third level. 
\end{theorem}
\begin{proof}
	Consider the following lossy GBS setup. We prepare $2K$ identical SMSV states with identical squeezing paramater $r$. These states are input, in pairs, into 50:50 beam splitters, generating two-mode squeezed vacuum (TMSV) states of the form 
	\begin{align}\label{eq:SMSV}
	\text{sech}(r)\sum_{S=0}^\infty  (\tanh r)^{S}\ket{2S,2S}~.
	\end{align}
	 For each TMSV state we couple one mode directly to a number-resolving detector (these are the heralding registers H), whereas the other half are sent into the lossy interferometer $\bm{U}$ (which may also require some additional vacuum inputs). The detectors at the output of $\bm{U}$ are called the boson sampling registers R. This entire setup is shown in Fig.~\ref{fig:exact-GBS-proof}.
	
	\begin{figure}[h]
		\includegraphics[scale=0.5]{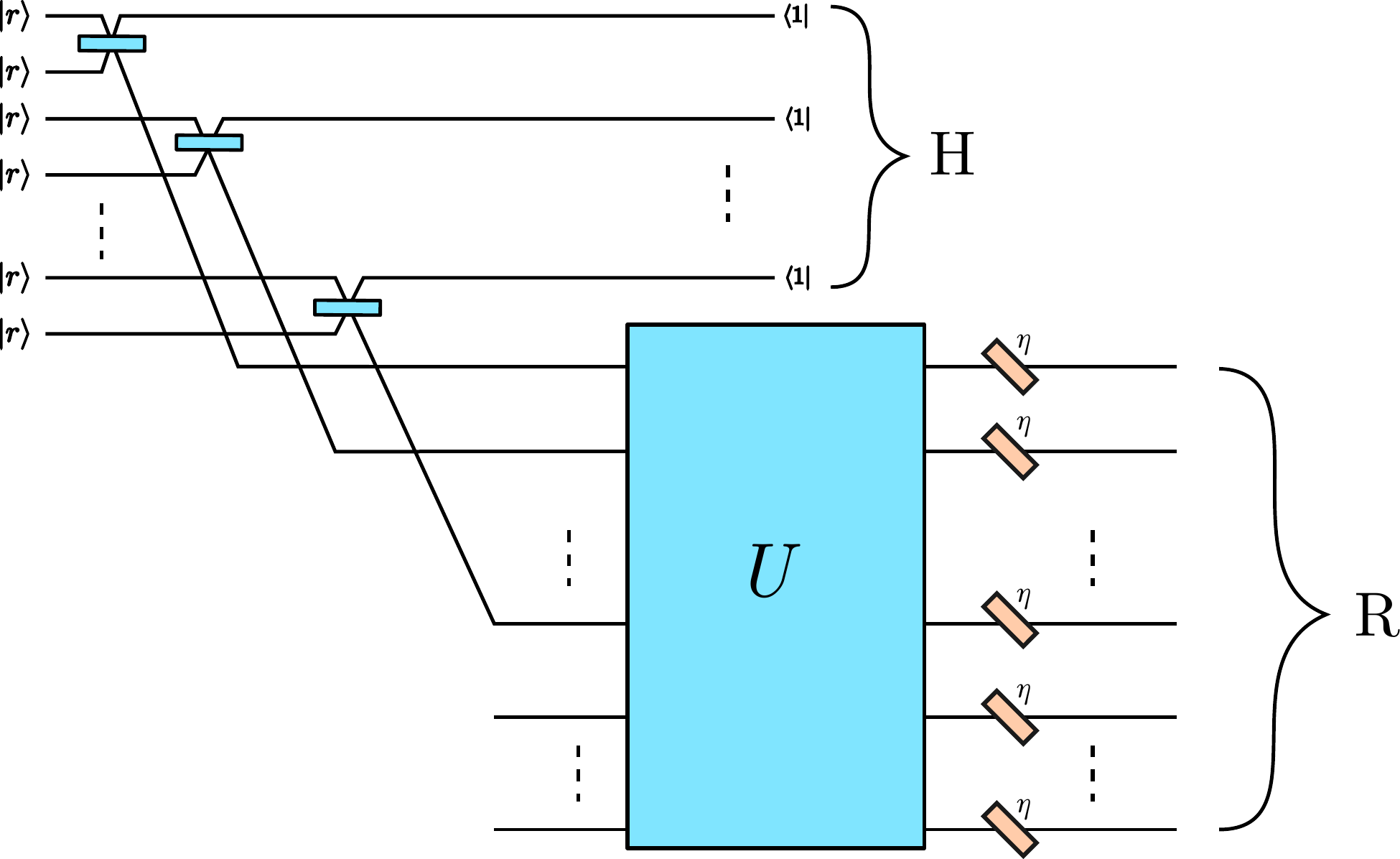}
		\caption{Lossy GBS setup used in the proof of Theorem~\ref{thm:exact-GBS}.}
		\label{fig:exact-GBS-proof}
	\end{figure}
	
	We now run this device, post-selecting on outcomes that satisfy two properties:
	\begin{itemize}
		\item[(i)] Exactly one photon is observed in each of the heralding modes H, and
		\item[(ii)] There are exactly $n$ photons in total in the R registers.
	\end{itemize}
	
	The two properties above guarantee that, in every event accepted by the post-selection, exactly one photon was injected into each non-vacuum input of $\bm{U}$ [due to the form of Eq.\ (\ref{eq:SMSV})], and no photons were lost within $\U$. Therefore, the resulting conditional probability distribution is the same as an ideal boson sampling instance with single-photon inputs and interferometer $\bm{U}$. 
	
	Now note that standard boson sampling, when augmented with the power of post-selection, can perform universal quantum computation \cite{aaronson2011computational}. Thus, by choosing the interferometer $\bm{U}$ properly, the same is true for the device of Fig.~\ref{fig:exact-GBS-proof}. From this it immediately follows that, by a standard argument (see e.g.\ \cite{bremner2010classical}), there can be no efficient classical algorithm to simulate the output distribution of the lossy GBS device exactly (or up to multiplicative error), otherwise the polynomial hierarchy collapses to its third level. 
\end{proof}

It is a well-accepted complexity-theoretic conjecture that the polynomial hierarchy is infinite, and so Theorem~\ref{thm:exact-GBS} can be taken as evidence that an efficient classical algorithm which exactly simulates a lossy GBS device does not exist. Note that the theorem did not require any assumptions on the strength of either losses or squeezing, so it holds for any squeezing parameter $r > 0$ and any interferometer transmissivity $\eta > 0$. Interestingly, if we replace $\bm{U}$ by the construction described in \cite{Brod2015depth}, Theorem \ref{thm:exact-GBS} also proves that GBS (lossy or not) is hard to simulate even if the entire linear-optical sector of Fig.~\ref{fig:exact-GBS-proof} has only five layers of (long-range) beam splitters.

Just like previous similar results~\cite{aaronson2011computational,rahimi2016sufficient}, Theorem~\ref{thm:exact-GBS} is not too relevant in a realistic scenario. The requirement of simulating the output distribution exactly (or with multiplicative error) is too strict, since a realistic device with experimental imperfections is not simulating the idealized device to that precision either. At best, Theorem~\ref{thm:exact-GBS} places bounds on how far a proposed efficient classical algorithm can be extended (for example, it shows that any exact simulation of lossy GBS based on the algorithm of \cite{rahimi2016sufficient} can only be efficient in the presence of dark counts).

\section{Mapping Squeezed lossy states to squeezed thermal states}
\label{sec:STS}
Any zero-mean single-mode Gaussian state can be decomposed into a squeezed thermal state $\rho=S(s_\rho,\phi_\rho)\rho_T^{n_\rho}S^\dagger(s_\rho,\phi_\rho)$~\cite{marian1993squeezed}. Here $\rho_T^{n_\rho}$ is a thermal state with average photon number $n_\rho$, and its covariance matrix is given by $(2n_\rho+1)\mb{I}_2$. $S(r,\phi)$ is the Stoler squeezing operator,
\begin{align}
S(s,\phi)=\exp\sbrac{\frac{1}{2}se^{i\phi}a^{\dagger 2}-\frac{1}{2}se^{-i\phi}a^{2}}.
\end{align}
Its transformation on the quadratures $x$ and $p$ is given by a symplectic matrix
\begin{align}
\begin{pmatrix}
\cosh r+\cos\phi\sinh r & -\sin\phi\sinh r\\
 -\sin\phi\sinh r & \cosh r-\cos\phi\sinh r 
\end{pmatrix}~.
\end{align} 
The covariance matrix of $\rho$ can be written as
\begin{align}\label{eq:cv-STS}
\mb{V}_\rho=\mb{V}_{STS}(s_\rho,n_\rho,\phi_\rho)=(2n_\rho+1)
\begin{pmatrix}
\cosh 2s_\rho-\cos\phi_\rho\sinh 2s_\rho & -\sin\phi_\rho\sinh 2s_\rho\\
-\sin\phi_\rho\sinh 2s_\rho & \cosh s_\rho+\cos\phi_\rho\sinh 2s_\rho 
\end{pmatrix}~.
\end{align}

Recall that a single-mode Gaussian state is $(t)$-classical when $\mb{V}_\rho-t\mb{I}_2$ is positive definite. From Eq.~\eqref{eq:cv-STS} we know that this happens when
\begin{align}
    s_\rho\leq \frac{1}{2}\ln\frac{2n_\rho+1}{t}~.
\end{align}

For a lossy squeezed state $\sigma$ with $\mb{V}_\sigma=\diag\brac{a_+,a_-}$ where $a_\pm=\eta e^{\pm2r}+(1-\eta)$, we can easily solve $s_\sigma$ and $n_\sigma$ by directly comparing $\mb{V}_\sigma$ with Eq.~\eqref{eq:cv-STS}, which gives
\begin{align}\label{eq-supp:sigma-s}
s_\sigma&=\frac{1}{4}\ln\frac{a_+}{a_-}~,\\
\label{eq-supp:sigma-n}
n_\sigma&=\frac{1}{2}(\sqrt{a_+a_-}-1)~.
\end{align}

\section{Sufficient conditions derived by using the quantum fidelity}
\label{sec:opt-fid}
{In the main text we derive the following bound by using the quantum fidelity,
\begin{align}\label{eq-supp:sufficient-condition-fid}
-\ln\sbrac{F_{\max}(\eta,q_D)}\leq \frac{\epsilon^2}{4K}~,
\end{align}
where
\begin{align}
F_{\max}(\eta,q_D):=\max_{t\in[1-2q_D,1]}\max_{\tau\in \C_{G}^{(t)}}F(\sigma,\tau)~.
\end{align}
We  explicitly calculate this quantity in what follows.

We start by computing the quantum fidelity between the lossy squeezed state $\sigma$ and a $(t)$-classical Gaussian state $\tau$. The fidelity between two single-mode Gaussian states is given by ~\cite{spedalieri2012limit}
\begin{align}\label{eq-supp:fidelity-Gaussian}
F(\sigma,\tau)=\frac{1}{\sqrt{\Delta+\Lambda}-\sqrt{\Lambda}}~,
\end{align}
where $\Delta=\frac{1}{4}\det(\mb{V}_\sigma+\mb{V}_\tau)$ and $\Lambda=\frac{1}{4}(\det \mb{V}_\sigma-1)(\det \mb{V}_\tau-1)$. To optimize Eq.~\eqref{eq-supp:fidelity-Gaussian} we make use of the STS parameterization for $\sigma$ and $\tau$. 
The fidelity is minimized when the squeezing axes are aligned \cite{marian1993squeezed}, so we can set $\phi_\tau=0$ for simplicity. Straightforward calculations then lead to
\begin{align}\label{eq:Delta}
\Delta&=(n_\sigma-n_\tau)^2+(2n_\sigma+1)(2n_\tau+1)\cosh^2(s_\sigma-s_\tau)~,\\\label{eq:Lambda}
\Lambda&=4n_\sigma(n_\sigma+1)n_\tau(n_\tau+1)~.
\end{align}

Since $\tau$ is a $(t)$-classical Gaussian state, using Eq.~(4) in the main text we have $s_\tau\leq s_0(\tau):=\frac{1}{2}\ln\frac{2n_\tau+1}{t}$. The task is then to find the point $(s^*_\tau,n_\tau^*)$ that maximizes the quantum fidelity subject to that constraint. Note from Eq.~\eqref{eq-supp:fidelity-Gaussian} that the fidelity monotonically decreases with $\abs{s_\sigma-s_\tau}$. So its optimization has two regimes.

First, when $s_\sigma\leq s_0(\sigma)$, a maximum of the fidelity is reached at $s_\tau^*=s_\sigma$ and $n_\tau^*=n_\sigma$, which gives $F(\s^*_\tau, n^*_\tau)=1$. This corresponds to the case when $\sigma \in \C_G^{(t)}$, i.e., $\sigma$ itself is a $(t)$-classical Gaussian state. This regime reproduces the previous result for exact simulation of GBS.

Second, when $s_\sigma\geq s_0(\sigma)$, the fidelity is maximized at $s_\tau^*=s_0(\tau)$. Substitute $s_0(\tau)$ into Eqs.~\eqref{eq-supp:fidelity-Gaussian}\--\eqref{eq:Lambda}, we have a function of  $n_\tau$ to optimize. It follows that its maximum is reached at $n_\tau^*=-\frac{1}{2}+\frac{1}{2}\sqrt{1+2t\sinh(2s_c)\exp(2s_\sigma)}$, where $s_c=\frac{1}{2}\ln(2n_\sigma+1)$.
The corresponding maximum fidelity is $F(s^*_\tau, n^*_\tau)=\text{sech}(s_\sigma-\frac{1}{2}\ln \frac{2n_\tau+1}{t})$. Our result agrees with previous result when $t=1$~\cite{marian2002quantifying}. 

Combining both regimes we write the maximum fidelity compactly as
$
\text{sech}\sbrac{\Theta\pbrac{s_\sigma-\frac{1}{2}\ln \frac{2n_\sigma+1}{t}}}
$,
where $\Theta$ is the ramp function. We now need to further optimize this over $t\in[1-2q_D,1]$. It is clear that the maximum value is achieved at $t = 1-2q_D$. Using Eqs.~\eqref{eq-supp:sigma-s}\--\eqref{eq-supp:sigma-n} we finally write the maximum fidelity as
\begin{align}
F_{\max}(\eta,q_D)=\text{sech}\sbrac{\frac{1}{2}\Theta\pbrac{\ln\pbrac{\frac{1-2q_D}{\eta e^{-2r}+1-\eta}}}}~.
\end{align}
Plugging this into Eq.~\eqref{eq-supp:sufficient-condition-fid} we obtain our final sufficient condition for classical simulability of noisy GBS. Notice that if we set $\epsilon=0$ we recover the corresponding condition for exact simulation given in Eq.~(8) in the main text, as expected.}

{
\section{Optimizing the simulability condition using sandwiched R\'{e}nyi relative entropy}
The additive error of our approximate classical algorithm of GBS is upper bounded by using the sandwiched R\'{e}nyi relative entropy, which is formally defined as follows for two quantum states $\rho_1$ and $\rho_2$~\cite{muller2013on,wilde2014strong},
\begin{equation}
 D_\alpha(\rho_1\Vert\rho_2)=\frac{1}{\alpha-1}\ln\tr\brac{\pbrac{\rho_1^{\frac{1-\alpha}{2\alpha}}\rho_2\rho_1^{\frac{1-\alpha}{2\alpha}}}^\alpha},
\end{equation}
for $\alpha\in(0,1)\cup (1,\infty)$. Below we list some of its properties which are useful for us later:
\begin{itemize}
    \item \textbf{Unitary invariance:} Given any unitary $U$, $D_\alpha(U\rho_1U^\dagger\Vert U\rho_2 U^\dagger)=D_\alpha(\rho_1\Vert\rho_2)$ for  $\alpha\in [\frac{1}{2},1)\cup (1,\infty)$~\cite{muller2013on}.
    \item \textbf{Additivity:} $D_\alpha(\rho_1\otimes w_1\Vert\rho_2\otimes w_2)=D_\alpha(\rho_1\Vert w_1)+D_\alpha(\rho_2\Vert w_2)$ for $\alpha\in [\frac{1}{2},1)\cup (1,\infty)$~\cite{muller2013on}.
    \item \textbf{Data processing:} For any completely positive trace-preserving map $\mathcal{E}$, we have
    \begin{align}
        D_\alpha(\rho_1\Vert\rho_2)\geq D_\alpha(\mathcal{E}(\rho_1)\Vert\mathcal{E}(\rho_2))~,
    \end{align}
    for $\alpha\in [\frac{1}{2},1)\cup (1,\infty)$~\cite{frank2013monotonicity,beigi2013sandwiched}.
    \item \textbf{Monotonicity:} $D_\alpha(\rho_1\Vert\rho_2)\geq D_{\alpha'}(\rho_1\Vert\rho_2)$ for $\infty>\alpha\geq\alpha'>0$~\cite{muller2013on,beigi2013sandwiched}.
\end{itemize}
Properties listed above makes the sandwiched R\'{e}nyi relative entropy a valid quantum distance measure. Moreover, it reduces to well-known quantum distance measures for specific values of $\alpha$. Specifically, we have~\cite{muller2013on,wilde2014strong}
\begin{align}
    D_{\frac{1}{2}}(\rho_1\Vert\rho_2)&=-\ln F(\rho_1,\rho_2),\\
    \lim_{\alpha\rightarrow 1} D_{\alpha}(\rho_1\Vert\rho_2)&= D(\rho_1\vert\rho_2):=\tr\brac{\rho_1(\ln\rho_1-\ln\rho_2)},\\
    \lim_{\alpha\rightarrow \infty} D_{\alpha}(\rho_1\Vert\rho_2)&:=D_{\max}(\rho_1\Vert\rho_2)=\inf\brac{\lambda\in\mathbb{R}: \rho_1\leq e^\lambda \rho_2},
\end{align}
which correspond to the logarithm of the quantum fidelity, the quantum relative entropy and the max-relative entropy, respectively.

Another essential ingredient in our deriviation of sufficient conditions is the generalized Pinsker's inequality~\cite{van2014renyi}:
\begin{align}
    D_T(P,Q)\leq \sqrt{\frac{2}{\alpha}D_\alpha(P\Vert Q)}~,
\end{align}
for $\alpha\in(0,1]$ and $D_\alpha(P\Vert Q)$ is the R\'{e}nyi divergence between two distributions,
\begin{align}
    D_\alpha(P\Vert Q)=\frac{1}{1-\alpha}\ln \sum_i p_i^\alpha q_i^{1-\alpha}~.
\end{align}
Notice that it is only proved for for $\alpha\in(0,1]$, which is the reason why we have to restrict our optimization over $\alpha\in [\tfrac{1}{2},1)$.

Similar to our calculation in the main text, by using the aforementioned properties, we can derive the following sufficient conditions for efficient simulation of GBS:
\begin{align}\label{eq:bound-alpha}
\frac{2}{\alpha}D^{\min}_\alpha(\eta,q_D)\leq \frac{\epsilon^2}{K}~,\quad \alpha\in[\frac{1}{2},1)~,
\end{align}
where $D^{\min}_\alpha(\eta,q_D)$ is the $\alpha$-order R\`{e}nyi relative entropy minimized over all permitted $(t)$-classical Gaussian states:
\begin{align}
D^{\min}_\alpha(\eta,q_D):= \min_{t\in[1-2q_D,1]}\min_{\tau\in \C_{G}^{(t)}}D_\alpha(\sigma\Vert\tau)~.
\end{align}
Since $D_\alpha$ is non-decreasing over $\alpha$, the l.h.s of Eq.~\eqref{eq:bound-alpha} is expected to reach it's minimum at some $\alpha^*$. To optimize over $\alpha$ we first try to calculate $D^{\min}_\alpha(\eta,q_D)$ for fixed $\alpha$. 

To facilitate our calculation we first define $Q(\sigma\Vert\tau)$ by
\begin{align}
D_\alpha(\sigma\Vert\tau)&:=\frac{1}{\alpha-1}\ln Q_\alpha(\sigma\Vert\tau)~,\\
Q_\alpha(\sigma\Vert\tau)  &=\tr\brac{\pbrac{\tau^{\frac{1-\alpha}{\alpha}}\sigma\tau^{\frac{1-\alpha}{\alpha}}}^\alpha}~.
\end{align}
From Ref.~\cite[Theorem 21]{seshadreesan2018renyi}, we have the following expression for $Q_\alpha(\sigma\Vert\tau)$ between two single-mode Gaussian states with zero mean:
\begin{align}
Q_\alpha(\sigma\Vert\tau)=\frac{1}{Z_\sigma^\alpha Z_\tau^{1-\alpha}}\sqrt{\det\sbrac{(\mb{V}_{\xi,\alpha}+i\mb{\Omega}})/2}~,
\end{align}
where
\begin{align}
\mb{Z}_\sigma&=\sqrt{\det\sbrac{(\mb{V}_{\sigma}+i\mb{\Omega})/2}}~,\\
\mb{V}_{\xi,\alpha}&=\frac{\pbrac{\mb{I}_2+(\mb{V}_\xi i\mb{\Omega})^{-1}}^\alpha  + \pbrac{\mb{I}_2-(\mb{V}_\xi i\mb{\Omega})^{-1}}^\alpha}
{\pbrac{\mb{I}_2+(\mb{V}_\xi i\mb{\Omega})^{-1}}^\alpha  - \pbrac{\mb{I}_2-(\mb{V}_\xi i\mb{\Omega})^{-1}}^\alpha }i\mb{\Omega}~,\\
\mb{V}_\xi & = \mb{V}_\sigma - \sqrt{\mb{I}_2+(\mb{V}_\sigma\mb{\Omega})^{-2}} \mb{V}_\sigma (\mb{V}_{\tau,\beta}+\mb{V}_\sigma)^{-1}\mb{V}_\sigma \sqrt{\mb{I}_2+(\mb{\Omega}\mb{V}_\sigma)^{-2}}~,\\
\beta&=\frac{1-\alpha}{\alpha}~.
\end{align}
From our optimization for $\alpha=1/2$, we expect the same landscape for general $\alpha$-order sanwiched R\'{e}nyi relative entropy: it is minimized at
\begin{align}\label{eq:landscape-1}
\phi_\tau&=\phi_\sigma~,\\
\label{eq:landscape-2}
s_\tau &= \begin{cases}
s_\sigma,\quad \text{for } s_\sigma < \frac{1}{2}\ln \frac{2n_\sigma+1}{\bar{t}}~,\\
\frac{1}{2}\ln\frac{2n_\tau+1}{\bar{t}} \quad \text{for  } s_\sigma \geq \frac{1}{2}\ln \frac{2n_\sigma+1}{\bar{t}}~.
\end{cases}
\end{align}
This will give us a function of $n_\tau$ to minimize.  The expression of $D_\alpha$ is too complicated to analytically show that above assumption is true. However, it can be verified analytically for $\alpha=1$, when the R\'{e}nyi relative entropy reduces to the quantum relative entropy~\cite{wilde2017gaussian}. 

For fixed $\alpha$, we obtain $D_\alpha^{\min}(\eta,q_D)$ by using Eq.~\eqref{eq:landscape-1} and then numerically minimizing over $n_\tau$. As shown in the Fig.~2 in the main text, We find out that the l.h.s of Eq.~\eqref{eq:bound-alpha} is minimized at $\alpha=\frac{1}{2}$. That is when $D_\alpha(\sigma,\tau)=-\ln F(\sigma,\tau)$. Therefore, the bound we calculate in the main text by using quantum fidelity is the tightest one we can get. To give an explicit example, in the figure above we plot $\frac{2}{\alpha}D_{\alpha}^{\min}$ against $\alpha$ for $\eta=0.1, r=0.11, q_D=0$.
\begin{figure}[H]
\centering
	\includegraphics[scale=0.5]{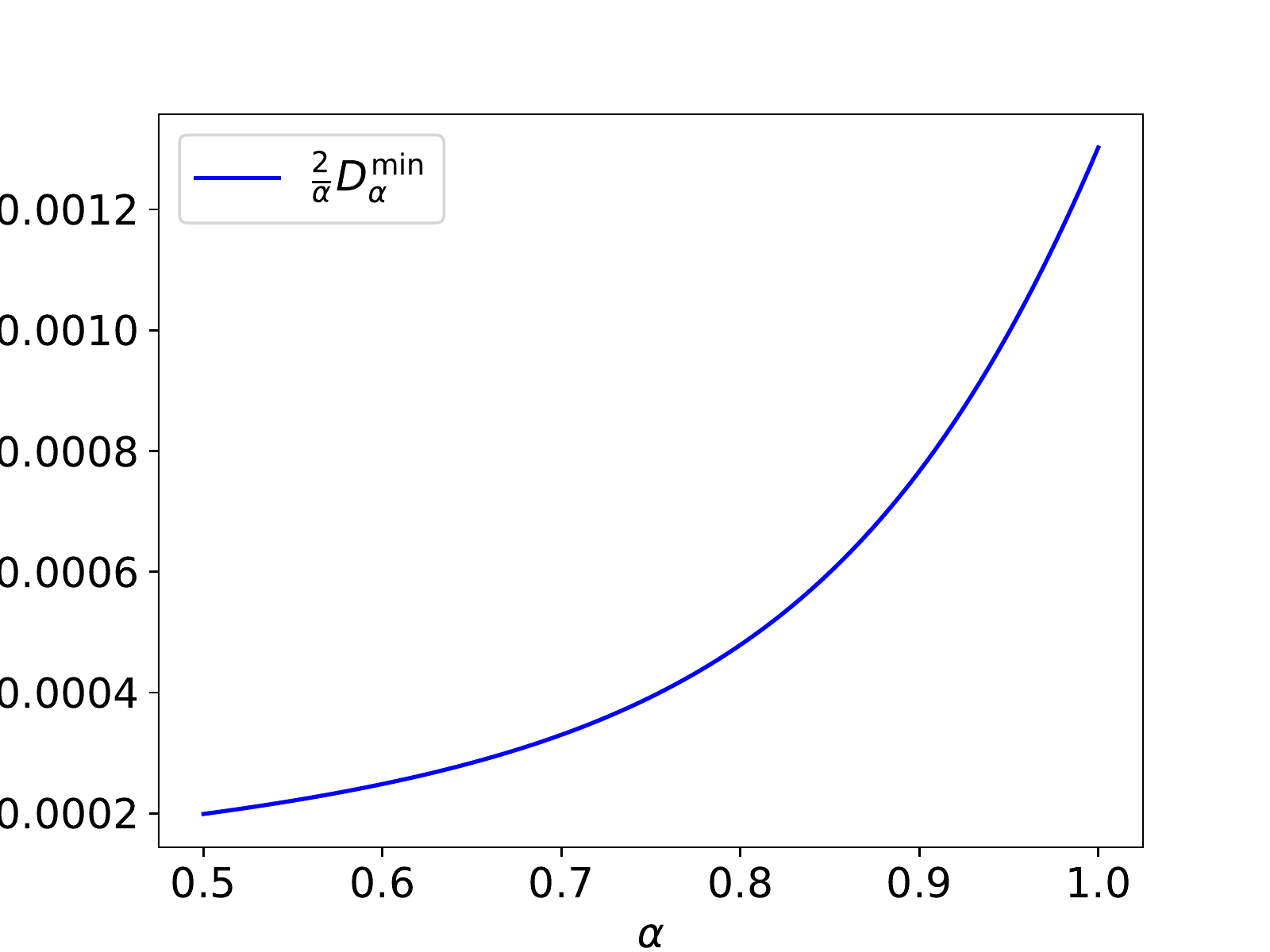}
	\caption{Here each point is obtained by numerically minimize $\frac{2}{\alpha} D_\alpha(\sigma\Vert\tau) $ over $\tau\in C_\G^{(t)} $ and over $t\in [1-2q_D, 1]$, assuming the landscape given in Eq.~\eqref{eq:landscape-1}-\eqref{eq:landscape-2}. The tightest bound happens at $\alpha=\frac{1}{2}$. \label{fig:distance-alpha}}
\end{figure}
\section{Efficient classical algorithm}
As explained in the main text, we devise our efficient approximate sampling algorithm in two steps: 1) for a given noisy Gaussian state find its closest classical state and 2) sample exactly from this classical state according to the algorithm given in \cite{rahimi2016sufficient}. Therefore, in order to present our approximate sampling algorithm, we first review the phase-space quasi-probability distribution (PQD) method used in \cite{rahimi2016sufficient}.  

Consider an $M$-mode bosonic system and let $\hat{\mb{x}}:=(\hat{q}_1,\hat{p}_1,\dots, \hat{q}_M,\hat{p}_M)$ denote the row vector of its quadrature operators. They satisfy the commutation relations ($\hbar=2$)
\begin{align}
\sbrac{\hat{x}_j,\hat{x}_k}& =2 i\Omega_{jk}~,\\
\mb{\Omega}
&=\mb{I}_M\otimes
\begin{pmatrix}
0 & -1 \\ 1 & 0
\end{pmatrix}~.
\end{align}
Here $\mb{I}_M$ is the $M\times M$ identity matrix. The annihilation operators are given by $\hat{a}_j=\tfrac{1}{2} \left( \hat{q}_j+i\hat{p}_j  \right)$ and  $\hat{\mb{a}}:=(\hat{a}_1, \ldots, \hat{a}_M)^T$.

The $\mb{t}$-ordered PQD [$(\mb{t})$-PQD] of an $M$-mode Hermitian operator $\rho$ is defined as 
\begin{align}
W^{(\mb{t})}_\rho(\mb{x})=(2\pi)^{-2M}\int d^{2M}\mb{\xi}\Phi^{(\mb{t})}_\rho(\mb{\xi})\exp(-i\mb{x}^T\mb{\Omega}\mb{\xi})~,
\end{align}
where $\mb{x}\in\mathbb{R}^{2M}$ are the eigenvalues of quadrature $\hat{\mb{x}}$ and $-\mb{I}_M\leq\mb{t}\leq\mb{I}_M$. Here $\Phi^{(\mb{t})}_\rho(\mb{\xi})=\tr[\rho D(\mb{\xi})]\exp(\mb{\xi}\mb{t}\mb{\xi}^\dagger/2)$ is the $\mb{t}$-ordered characteristic function and $D(\mb{\xi})=\exp(i\hat{\mb{x}}^T\mb{\Omega}\mb{\xi})$ is the displacement operator. We slightly abuse the notation that  $\mb{t}$ is a diagonal matrix with various order parameters on its diagonal. For $\mb{t} =
-\mb{I}_M$, $\mb{t}=0$ and $\mb{t}=\mb{I}_M$, we obtain the Husimi, Wigner and Glauber-Sudarshan functions, respectively. 

In Ref.~\cite{rahimi2016sufficient} it was shown that the probability of detecting a photon pattern can be written as the overlap between PQDs:
\begin{align}\label{eq:sampling-chain}
    P(\mb{n})= (2\pi)^M\int d^{2M}\mb{x}W_{\Pi}^{(-\mb{t})}(\mb{n}\vert\mb{x})W_{\rho_{\text{out}}}^{(\mb{t})}(\mb{x})~,
\end{align}
where $W_{\rho_{\text{out}}}^{(\mb{t})}(\mb{x})$ is the $(\mb{t})$-PQD of the pre-measurement state and $W_{\Pi}^{(-\mb{t})}(\mb{n}\vert\mb{x})$ is the $(-\mb{t})$-PQD of the measurement operator. The crucial observation in Ref.~\cite{rahimi2016sufficient} is the following: If both PQDs are positive and can be simulated efficiently for some $\mb{t}$, the device as a whole can be efficiently simulated by successively sampling from the chain of distributions given in Eq.~\eqref{eq:sampling-chain}. 

For threshold detector with quantum efficiency $\eta_D$ and dark count rate $q_D$, its POVM elements are given by
\begin{align}
\Pi_0&=(1-p_D)\sum_{n=0}^\infty(1-\eta_D)^n\ket{n}\bra{n}~,\\
\Pi_1&=I-\Pi_0~.
\end{align}
It was shown that their PQDs
\begin{align}\label{eq:PQD-detector}
    W_\Pi^{(-t)}(0|\mb{x})&=\frac{1-p_D}{2\pi[1-\eta_D(1-t)/2]}\exp\brac{-\frac{\eta_D\abs{\mb{x}}^2}{4[1-\eta_D(1-t)/2]}}~,\\
    W_\Pi^{(-t)}(1|\mb{x})&= \frac{1}{2\pi}-W_\Pi^{(-t)}(0|\mb{x})
\end{align}
are non-negative provided $t\geq \bar{t}=1-2p_D/\eta_D$~\cite{rahimi2016sufficient}.
Since $(\mb{t})$-classical Gaussian states have positive Gaussian $(t\mb{I}_M)$-PQDs,
\begin{align}\label{eq:PQD-Gaussian}
W^{(\mb{t})}_\rho(\mb{x})=
\frac{\exp\sbrac{-\frac{1}{2}\mb{x}^T(\mb{V}_\rho-t\mb{I}_{2M})^{-1}\mb{x}}}{(2\pi)^M\sqrt{\det(\mb{V}_\rho-t\mb{I}_{2M})}}~,
\end{align}
any $(\bar{t}\mb{I}_M)$-classical Gaussian states can be exactly and efficiently sampled according to ~\eqref{eq:sampling-chain}. Among those state, the one which is closest to the output state of a noisy GBS device is solved in the main text and with more details in Sec.~\ref{sec:opt-fid}.

We outline our approximate sampling algorithm in Alg.~\ref{alg:GBS}. Since the construction of the covariance matrix of a multi-mode Gaussian state, sampling from a multi-variant Gaussian distribution and sampling from a Bernoulli distribution can be done efficiently, our algorithm can be executed in polynomial time. An implementation using the Strawberry Fields library~\cite{killoran2019strawberry} can be found at~\footnote{Github repository at \url{https://github.com/dionysos137/GBS_classicality}}.
\begin{algorithm}[H]
\caption{Efficient approximate sampling of GBS}\label{alg:GBS}
\begin{algorithmic}[1]
\Procedure{ClosestClassicalState}{$r,K,\eta,\bm{U},t$}
    \State $M\gets$ size of $\bm{U}$ \Comment{Number of modes}
    \State $a_\pm\gets \eta e^{\pm 2r}+1-\eta$ \Comment{Recall $
\rho_{\text{out}}=\mathcal{U}\left(\sigma^{\otimes K}\otimes \ket{0}\bra{0}^{\otimes (M-K)}\right)\mathcal{U}^\dagger$}
    \State $s_\sigma\gets \frac{1}{4}\ln\frac{a_+}{a_-}$ \Comment{STS parameters of $\sigma$, see Sec.~\ref{sec:STS}}
    \State $n_\sigma\gets \frac{1}{2}(\sqrt{a_+a_-}-1)$
   
    \If{$s_\sigma <\frac{1}{2}\ln\frac{2n_\sigma+1}{t}$}\Comment{Find the closest state $\tau$, see Sec.~\ref{sec:opt-fid}}
        \State $s^*_\tau\gets s_\sigma$
        \State $n^*_\tau\gets n_\sigma$
    \Else

         \State $s_c\gets\frac{1}{2}\ln(2n_\sigma+1)$
    \State $n_\tau\gets-\frac{1}{2}+\frac{1}{2}\sqrt{1+2t\sinh(2s_c)\exp(2s_\sigma)}$
    \State $s^*_\tau\gets \frac{1}{2}\ln\frac{2n_\tau+1}{t}$
    \EndIf
    \State $\mb{V}_{\tilde{\rho}}\gets$ covariance matrix of $
\tilde{\rho}=\mathcal{U}\left(\tau^{\otimes K}\otimes \ket{0}\bra{0}^{\otimes (M-K)}\right)\mathcal{U}^\dagger$
    \State \textbf{Return} $\mb{V}_{\tilde{\rho}}$
\EndProcedure

\Procedure{SampleGBS}{$r,K,\eta,\bm{U},\eta_D, q_D,\epsilon$}
\State $M\gets$ size of $\bm{U}$ \Comment{Number of modes}
\State $\bar{t}\gets 1-2p_D/\eta_D$ \Comment{Will sample from a $(\bar{t})$-classical Gaussian state}
\If{$\text{sech}\sbrac{\frac{1}{2}\Theta\pbrac{\ln\pbrac{\frac{\bar{t}}{\eta e^{-2r}+1-\eta}}}} < e^{-\frac{\epsilon^2}{4K}}$}\Comment{Inequality (1) in the main text} 
    \State \textbf{Return} \Comment{Our test is passed, program aborts} 
\Else
    \State $\mb{V}\gets$ \textsc{ClosestClassicalState}$(r,\eta,U,\bar{t},K)$
    \State $\mb{x}\gets$\textsc{SampleNormal}$(mu=0, cov=(\mb{V}-\bar{t}\mb{I}_M)^{-1}))$ \Comment{See Eq.~\eqref{eq:PQD-Gaussian}}
    \State $\mb{n}\gets$ empty array of length $M$
    \For{each $i\in\sbrac{1,M}$}
        \State $\mb{x}_i\gets(\mb{x}[i],\mb{x}[i+1])$
        \State $\mb{n}[i]\gets$\textsc{SampleBernoulli}$(p=2\pi W_\Pi^{(-\bar{t})}(0\vert\mb{x}_i))$
        \Comment{See Eq.~\eqref{eq:PQD-detector}}
    \EndFor
    \State \textbf{Return} $\mb{n}$
\EndIf
\EndProcedure
\end{algorithmic}
\end{algorithm}
\end{document}